\documentclass[a4paper]{article}
\usepackage{mathtools}

\usepackage{enumitem}
\usepackage{calc}
\usepackage{comment}

\usepackage{lmodern} 

\usepackage{graphicx} 

\usepackage{amsmath}
\usepackage{amssymb}
\usepackage{pict2e}
\usepackage{amsthm}
\usepackage{amscd}
\usepackage{authblk}
\usepackage{mathrsfs}
\usepackage{todonotes}
\usetikzlibrary{calc} 

\usepackage{pgf,tikz}
\usetikzlibrary{arrows,angles,quotes,decorations.markings}
\tikzstyle directed=[postaction={decorate,decoration={markings,
		mark=at position .65 with {\arrow{latex}}}}]
\usepackage{yfonts}

\usepackage{marvosym}
\usepackage[percent]{overpic}

\newtheorem{theorem}{Theorem} [section]
\newtheorem{proposition}[theorem]{Proposition}	
	
\newtheorem{lemma}[theorem]{Lemma}

\newtheorem{pb}[theorem]{Riemann-Hilbert Problem}

\newtheorem{remark}[theorem]{Remark}

\theoremstyle{definition}
\newtheorem{definition}[theorem]{Definition}

\DeclareMathOperator{\tr}{Tr}

\newcommand{\R}{\mathbb{R}}

\newcommand{\Ai}{{\rm Ai}}

\def\XXint#1#2#3{{\setbox0=\hbox{$#1{#2#3}{\int}$}
		\vcenter{\hbox{$#2#3$}}\kern-.5\wd0}}

\usepackage{tikz}
\usetikzlibrary{arrows}
\usetikzlibrary{decorations.pathmorphing}
\usetikzlibrary{decorations.markings}
\usetikzlibrary{patterns}
\usetikzlibrary{automata}
\usetikzlibrary{positioning}
\usepackage{tikz-cd}
\tikzset{->-/.style={decoration={
			markings,
			mark=at position #1 with {\arrow{latex}}},postaction={decorate}}}

\tikzset{-<-/.style={decoration={
			markings,
			mark=at position #1 with {\arrowreversed{latex}}},postaction={decorate}}}

\usetikzlibrary{shapes.misc}\tikzset{cross/.style={cross out, draw, 
		minimum size=2*(#1-\pgflinewidth), 
		inner sep=0pt, outer sep=0pt}}

\usepackage{pgfplots}

\numberwithin{equation}{section}

\def\ds{\displaystyle}

\usepackage[colorlinks=true]{hyperref}
\hypersetup{urlcolor=blue, citecolor=red, linkcolor=blue}

\newcommand{\be}{\begin{equation}}
	\newcommand{\ee}{\end{equation}}

\def\e{{\epsilon}}
\def\ds{\displaystyle}

\def\Ai{{\rm Ai \,}}
\def\I{{\rm I \,}}

\def\P2n{{\rm P}_{{\rm II}}^{(n)}}
\def\R{\mathbb{R}}

\def\i{ {\mathrm{i}}}
\def\e{{\mathrm{e}}}
\def\1{\mathbf{1}}

\def\R{\mathbb{R}}
 
\def\d{\mathrm{d}}

\def\Ai{\mathrm{Ai}}

\def\x{{\vec{x}}}

\def\diag{\mathrm{diag}}
\def\I{\mathrm{I}}
\def\sgn{{\mathrm{sgn}}}

\begin{document}
	\title{The Riemann-Hilbert approach to the generating function of the higher order Airy point processes}
	\author[1]{Mattia Cafasso}
	\author[2]{Sofia Tarricone}
	\renewcommand\Affilfont{\small}
	\affil[1]{\textit{LAREMA, UMR 6093, UNIV Angers, CNRS, SFR Math-STIC, France;} \texttt{cafasso@math.univ-angers.fr}}
	\affil[2]{\textit{Institut de Recherche en Math\'ematique et Physique,  UCLouvain, Chemin du Cyclotron 2, B-1348 Louvain-la-Neuve, Belgium;} \texttt{sofia.tarricone@uclouvain.be}}

	\date{}
	\maketitle
	
	\begin{abstract}
		We prove a Tracy-Widom type formula for the generating function of occupancy numbers on several disjoint intervals of the higher order Airy point processes. The formula is related to a new vector-valued Painlevé II hierarchy we define, together with its Lax pair.
		
	\end{abstract}
	
	
	\section{Introduction}
	
	Let us consider the higher order Airy functions
	
	\begin{equation}
		\Ai_n(x) := \frac{1}{\pi} \int_{0}^\infty \cos\left(\frac{y^{2n + 1}}{2n + 1} + xy \right) \mathrm{d}y,\,  \quad x \in \mathbb R, \, n \in \mathbb N,
	\end{equation}
	and the associated kernels
	\begin{equation}\label{eqintro:kernel}
		K_n(x,y) := \int_{0}^\infty \Ai_n(x + z)\Ai_n(y + z)\mathrm{d}z.
	\end{equation}
	It is easy to prove, using standard arguments in the the theory of point processes (Theorem 3 in \cite{Sosh:RandomPointFields}) that the kernels $K_n$, for any $n \geq 1$, define a determinantal point process whose correlation functions are given by the standard formula
	$$
	\rho_{\ell;n}(x_1,\ldots,x_\ell) := \det \Bigg(K_n(x_i,x_j)\Bigg)_{i,j = 1}^\ell \quad \ell \geq 1,
	$$ 
	see Appendix A in \cite{CafassoClaeysGirotti}. 
	The importance of these point processes stems from applications to statistical physics and combinatorics. Indeed, they are associated to new universality classes generalizing the KPZ one (case $n = 1$). These universality classes describe both the limiting behavior of the momenta of non--interacting fermions trapped in an anharmonic potential \cite{FermMM} and the one of multicritical random partitions \cite{BeteaBouttierWalsh,KimuraZahabi, KimuraZahabi2}.\\ 
	Let us denote with 
	$$
	\zeta_1^{(n)} > \zeta_2^{(n)}  > \zeta_3^{(n)}  > \ldots > \zeta_j^{(n)} > \ldots
	$$
	the (random) points in the process, fix a collection $\{A_j, j = 1, \ldots, k\}$ of intervals of the form
	$$
	A_j = (x_j,x_{j-1}), \quad \mathrm{with} \; + \infty =: x_0 > x_1 > \ldots > x_k > -\infty
	$$
	and some real constants $\alpha_1,\ldots,\alpha_k$ such that $\alpha_j \in [0,1)$ for all $j$. We will denote $$\#_{A_j}^{(n)} := \#\{\zeta_\ell^{(n)} | \zeta_\ell^{(n)} \in A_j ,\; \ell \geq 0\}$$ the random variable counting the number of points contained in the interval $A_j$. We are interested in studying the generating function
	\begin{equation}\label{def:genfunction}
		F_n(\vec{x},\vec{\alpha}) = F_n(x_1,\ldots,x_k;\alpha_1,\ldots,\alpha_k) := \mathbb E \left[\prod_{j = 1}^k (1 - \alpha_j)^{\#_{A_j}^{(n)}} \right],
	\end{equation}
	whose derivatives give the joint probability law of $k$ given particles in the process. More precisely, given $m_1 < \ldots < m_k$,
	\begin{equation}
		\mathbb P \left(\bigcap_{j = 1}^k \left(\zeta_{m_j}^{(n)} < x_j \right) \right) = \sum \frac{(-1)^{j_1 + \ldots j_k}}{j_1!j_2!\ldots j_k!} \frac{\partial^{j_1 + j_2 + \ldots +j_k}}{\partial \alpha_1^{j_1}\partial \alpha_2^{j_2}\ldots\partial \alpha_k^{j_k}}F_n(\vec x, \vec \alpha)\Big|_{\vec \alpha = (1,\ldots,1)},
	\end{equation}
	where the sum is taken over all indices $j_1,\ldots,j_k$ satisfying the conditions
	$$
	j_1 < m_1, \; j_1 + j_2 < m_2, \cdots, \sum_{\ell = 1}^k j_\ell < m_k.
	$$
	(see for instance \cite{BaikDeiftRains}).
	The main result of this paper is a Tracy--Widom formula for $F_n(\vec{x},\vec{\alpha})$, relating the latter to a vector--valued version of the Painlevé II hierarchy we are going to define. Our formula generalizes both the one obtained by Claeys and Doeraene \cite{ClaeysDoeraene} for the case $n = 1$ and arbitrary $k \geq 1$, and the one obtained by one of the authors, Claeys and Girotti for arbitrary $n \geq 1$  and $k = 1$ \cite{CafassoClaeysGirotti}.\\
	
	In order to state precisely our result, we need to introduce some (vector--valued) differential polynomials which will be used to define our hierarchy of equations. We will work with the ring $$\mathcal R := \mathbb C[u_1,\ldots,u_k,Du_1,\ldots, Du_k,D^2u_1,\ldots D^2 u_k,\ldots]$$ generated by $k$ functions $u_j : \mathbb R \ni t \mapsto u_j(t), \; j = 1,\ldots, k$ and its derivatives, and denote with $D^{-1}$ the left--inverse of the derivation,  such that $D^{-1}D v = v$ for all $v$ in $\mathrm{Im}(D)$. Given $\vec v,\vec w \in \mathcal R^k$, let us define 
	$$
	<\vec v, \vec w> := \vec v^\top \vec w \in \mathcal R, \quad \{\vec v,\vec w\} := \vec v \vec w^\top + \vec w \vec v^\top \in \mathrm{Mat}(k,\mathcal R), \quad [\vec v,\vec w] := \vec v \vec w^\top - \vec w \vec v^\top \in \mathrm{Mat}(k,\mathcal R).
	$$
	We will also denote $\vec u := (u_1,\ldots,u_k)^\top \in \mathcal R^k$.
	\begin{definition}
		Suppose that $\vec v \in \mathcal R^k$ is such that 
		$$<\vec u,\vec v>\, \in D(\mathcal R) \quad \text{and} \quad \{\vec u,\vec v\} \in D(\mathrm{Mat}(k,\mathcal R)).$$ We define
		$$
		\mathcal L^{\vec u}_+ \vec v := \i\, D \vec v - \i \left(D^{-1} \{\vec u,\vec v\}\right) \vec u - 2 \i \left(D^{-1}<\vec u, \vec v >\right) \vec u.
		$$
		Analogously, for any $\vec v$ such that $ [\vec u,\vec v] \in D(\mathrm{Mat}(k,\mathcal R))$, we define
		$$
		\mathcal L^{\vec u}_- \vec v := \i\, D \vec v + \i \left(D^{-1} [\vec u,\vec v]\right) \vec u.
		$$
	\end{definition}
	\begin{theorem}\label{thm:main}
		Let $F_n(\vec x,\vec \alpha)$ defined as in \eqref{def:genfunction} with $\alpha_j \neq \alpha_{j+1}$ for all $j \leq k-1$, and $\alpha_{k+1} \equiv 0$. Then
		\begin{equation}\label{TWformula}
			F_n(\vec x, \vec \alpha) =\exp\left( - \int_{0}^\infty t <\vec u(t),\vec u(t)> \mathrm d t \right),
		\end{equation}
		where $\vec u(t) = \vec u(n,\vec x + t,\vec \alpha)$ satisfy the following (vector-valued) ordinary differential equation
		\begin{equation}\label{PIIequationIntro}
			\Big(\mathcal L^{\vec u}_+\mathcal L^{\vec u}_-\Big)^n \vec u(t) = -\diag \big(x_1 + t,\ldots,x_k +t \big) \vec u(t)
		\end{equation}
		and have the following behavior at $+ \infty$
		\begin{equation}\label{asymptoticsIntro}
			\vec u(n, \vec x + t, \vec \alpha) = \Bigg(\sqrt{\alpha_j - \alpha_{j+1}}\Ai_n(t + x_j)(1 + o(1))\Bigg)_{j = 1,\ldots,k}.
		\end{equation}
	
	Moreover, if $\alpha_{j+1} < \alpha_j$ then $u_j(n,\vec x + t, \vec \alpha)$ is real-valued for real $t$. If $\alpha_{j+1} > \alpha_j$, then  $u_j(n,\vec x + t, \vec \alpha)$ is purely imaginary for real $t$.
	\end{theorem}
	\begin{remark}
		We will call the collection of equations in \eqref{PIIequationIntro} the \emph{vector valued PII hierarchy}. Their formulation, from an algebraic point of view, is completely analogous to the one we previously introduced, in collaboration with Thomas Bothner \cite{BothnerCafassoTarricone}, for the integro-differential Painlevé II hierarchy, see also \cite{Kraj}. Note, however, that the results contained in this article cannot be deduced from the ones in \cite{BothnerCafassoTarricone}, because of the assumption of smoothness for the weight function $w$, see Section 1.3 in loc.cit.
	\end{remark}
\begin{remark} We write down explicitly the first two members of the hierarchy \eqref{PIIequationIntro}, using the shorthand notation $\dot{\vec{u}} =D\vec u$ to denote the derivative.
For $n=1$, equation \eqref{PIIequationIntro} is 
\begin{equation}\label{eq:vectorPIIn1}
\ddot{\vec{u}} = 2 \vec u \vec u^{\top} \vec u + (\vec{x} + t ) \vec u \;\; \text{i. e.}\;\; \begin{cases}
\ddot{u}_1= 2 u_1\sum_{j=1}^{k}u_j^2 + (t+x_1)u_1\\
\ddot{u}_2=2u_2\sum_{j=1}^ku_j^2 + (t+x_2)u_2\\
\vdots\\
\ddot{u}_k=2u_k\sum_{j=1}^ku_j^2 + (t+x_k)u_k
\end{cases}
\end{equation}
that coincide indeed with the coupled system of Painlevé II equations introduced in \cite{ClaeysDoeraene}. For $n=2$, equation \eqref{PIIequationIntro} is 
\begin{equation}
\ddddot{\vec{u}} = 4 \ddot{\vec u}  \vec u^{\top} \vec u+ 8 \dot{\vec u } \dot{ \vec u}^{\top} \vec u+ 6 \vec u \vec u^{\top} \ddot{\vec u}+2u \dot{\vec u}^{\top}\dot{\vec u} -6 \vec u (\vec u^{\top}\vec u)^2 -(t+\vec{x})\vec u
\end{equation}
which is indeed as a vector-valued version of the second member of the Painlevé II hierarchy.
\end{remark}
The paper is organised as follows: in Section \ref{sec:2} we prove that the generating function $F_n(\vec{x}+t, \vec{\alpha})$ is equal to the Fredholm determinant of an integrable operator of IIKS type \cite{IIKS}. As a byproduct, we formulate and use the Riemann-Hilbert problem \ref{pb:main rhpb} to compute the logarithmic derivative of $F_n(\vec{x}+t, \vec{\alpha})$ with respect to $t$, and this concludes Section \ref{sec:2}. In Section \ref{sec3} we associate to the Riemann-Hilbert problem \ref{pb:main rhpb} a Lax pair for the vector-valued Painlevé II hierarchy \eqref{PIIequationIntro}. Section \ref{sec:4} concludes, collecting all the previous results, the proof of Theorem \ref{thm:main}.
	
	\section{$F_n(\vec{x},\vec{\alpha})$ and the associated Riemann-Hilbert problem}
	\label{sec:2}
	It is well known (see for instance \cite{Sosh:RandomPointFields}) that the generating function $F_n(\vec x, \vec \alpha)$ defined in \eqref{def:genfunction} can be expressed as a Fredholm determinant. More precisely,
	\begin{equation}\label{eq:FredhDet}
		F_n(\vec x, \vec \alpha) = \det \left(\I - \sum_{j = 1}^k \alpha_j \mathbb K_{n| A_j}\right),
	\end{equation}
	where $\mathbb K_n$ is the integral operator associated to the kernel \eqref{eqintro:kernel} and, for any Borel subset $B \subseteq \mathbb R$,  $\mathbb K_{n| B}$ indicates the restriction of $\mathbb K_n$ to $B$. For our purposes, it is convenient to recall a different representation of the kernel $K_n$ as a double contour integral. In what follows, let us denote
	\begin{equation}
		\psi_n(\lambda;t) := \frac{\lambda^{2n + 1}}{2n + 1} + \lambda t.
	\end{equation}
	It is easy to show (see for instance \cite{BothnerCafassoTarricone}) that, for any real $t$,
	\begin{equation}\label{defalternAi}
		\Ai_n(t) = \frac{1}{2 \pi} \int_{\Gamma_+} \exp \Big(\i \psi_n(t;\lambda) \Big) \d \lambda =  \frac{1}{2 \pi} \int_{\Gamma_-} \exp \Big(-\i \psi_n(t;\lambda) \Big) \d \lambda,
	\end{equation}
	where $\Gamma_+$ is any smooth contour oriented from $\infty{\rm e}^{\i a}$ to $\infty{\rm e}^{\i b}$ with $a \in \left(\frac{2n \pi}{2n + 1}, \pi \right)$ and $b \in (0, \frac{\pi}{2n + 1})$ (see Fig. \ref{fig:curves}), and $\Gamma_-$ its reflection with respect to the real axis. Actually, one can take \eqref{defalternAi} as an alternative definition of $\Ai_n$ (or, rather, of its analytical continuation). Then, combining \eqref{eqintro:kernel} with \eqref{defalternAi}, one proves \cite{CafassoClaeysGirotti, BothnerCafassoTarricone}
	\begin{lemma}
		The kernel defined in \eqref{eqintro:kernel} admits the double--contour integral representation
		\begin{equation}\label{eq:doublecontourkernel}
			K_n(x,y) = \frac{\i}{(2 \pi)^2} \int_{\Gamma_+} \d\lambda \int_{\Gamma_-} \d\mu \; \frac{{\rm e}^{\i (\psi_n(\lambda;x) - \psi_n(\mu;y))}}{\lambda - \mu} .
		\end{equation}
	\end{lemma}

\begin{figure}
	\centering
	\begin{tikzpicture}[scale=0.6]
	\draw[very thick,->] (-5.5,0) -- (5.5,0);
	\draw[very thick,->] (0,-1) -- (0,4.5);
	\draw[fill=gray, dashed, opacity=0.5] (5,0) to[out=90, in=-30] (3.2,2.6) node[above,black]{$\frac{\pi}{2n+1}$} to (0,0) to (-5,0) to[out=90, in=-150] (-3.2,2.6) node[above,black]{$\frac{2n\pi}{2n+1}$} to (0,0);
	\draw[>=latex,directed,blue] (-4, 0.3) to (4,0.3); 
	\draw[blue,dashed] (-5.5,.3) to (-4, 0.3);
	\draw[blue,dashed] (5.5,.3) to (4, 0.3);
	\draw[>=latex,directed,blue] (-4, 2)..controls (0,0.000000001).. (4,2) node[below right,blue]{$\Gamma_{+}$};
	\draw[dashed,blue] (-5, 2.5) to (-4,2);
	\draw[dashed,blue] (4, 2) to (5,2.5);
	\end{tikzpicture}
	\caption{These are possible choices for the curve $ \Gamma_+$ appearing the integral representation  of the Airy function $\mathrm{Ai}_n$. }
	\label{fig:curves}
\end{figure}
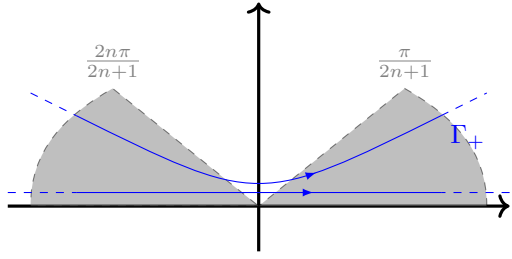
	We now define two vector--valued functions\footnote{The notation we used in \eqref{eq: f,g} reflects the fact that it will be convenient, in the sequel, to think about $f_n,g_n$ as functions taking values in $\mathbb R \oplus \mathbb R^k$} $f_n,g_n = \Gamma \longrightarrow \R^{k+1}$, with $\Gamma := \Gamma_+ \cup \Gamma_-$:
	\begin{equation}\label{eq: f,g}
		f_n(\lambda) := \frac{1}{2 \pi} \begin{pmatrix}
			\e^{-\frac{\i}2 \psi_n(\lambda;0)} \chi_{\Gamma_-}(\lambda) \\
			\\
			\hline
			\\
			\sqrt{\alpha_1 - \alpha_2}\e^{\frac{\i}2 \psi_n(\lambda; 2t + 2x_1) } \chi_{\Gamma_+}(\lambda)\\
			\vdots\\
			\\
			\sqrt{\alpha_k - \alpha_{k+1}}\e^{\frac{\i}2 \psi_n(\lambda; 2t + 2x_k) } \chi_{\Gamma_+}(\lambda)\
		\end{pmatrix}, \quad 
		g_n(\lambda) := \begin{pmatrix}
			\e^{\frac{\i}2 \psi_n(\lambda;0)} \chi_{\Gamma_+}(\lambda) \\
			\\
			\hline
			\\
			\sqrt{\alpha_1 - \alpha_2}\e^{-\frac{\i}2 \psi_n(\lambda; 2t + 2x_1) } \chi_{\Gamma_-}(\lambda)\\
			\vdots\\
			\\
			\sqrt{\alpha_k - \alpha_{k+1}}\e^{-\frac{\i}2 \psi_n(\lambda; 2t + 2x_k) } \chi_{\Gamma_-}(\lambda)\
		\end{pmatrix}
	\end{equation}
	and denote with $\mathbb L_n : L^2(\Gamma) \longrightarrow L^2(\Gamma)$ the associated integral operator with integrable (in the sense of \cite{IIKS}) kernel defined by the equation
	\begin{equation}\label{eq:Lintegrable}
		(\lambda - \mu)L_n(\lambda,\mu) = f_n^\top(\lambda)g_n(\mu).
	\end{equation}

	The following proposition is a generalization of Proposition 2.1 in \cite{CafassoClaeysGirotti} (see also \cite{BertolaCafasso1} for the case $n = 1$ and $k$ arbitrary). 
	\begin{proposition} The generating function $ F_n(\vec{x} + t,\vec{\alpha})$ coincide with the Fredholm determinant of the integrable operator $\mathbb{L}_n$, i.e.
		\begin{equation} 
			\det(\I - \mathbb L_{n}) = F_n(\vec{x} + t,\vec{\alpha}).
		\end{equation}
	\end{proposition}
	\begin{proof}
		From the very definition of $\mathbb L_n$, using the natural polarization of $L^2(\Gamma) \simeq L^2(\Gamma_+) \oplus L^2(\Gamma_-) $, we can write in block form
		$$
		\I - \mathbb L_n = \begin{pmatrix} 
			\I & -\mathbb F_n \\
			-\ds\sum_{j = 1}^k (\alpha_j - \alpha_{j + 1}) \mathbb G_{n;x_j} & \I
		\end{pmatrix}
		$$
		where $\mathbb F_n : L^2(\Gamma_+) \longrightarrow L^2(\Gamma_-), \quad \mathbb G_{n;x_j} : L^2(\Gamma_-) \longrightarrow L^2(\Gamma_+), \; j = 1, \cdots, k$ and $\mathbb L_n,\mathbb G_{n;x_j}$ have kernels given by
		$$
		(\mu - \lambda)F_n(\mu,\lambda) = \frac{1}{2 \pi}{\rm e}^{\frac{i}2 (\psi_n(\lambda;0) - \psi_n(\mu;0))}\chi_{\Gamma_-(\mu)}\chi_{\Gamma_+(\lambda)},
		$$
		$$
		(\xi - \mu)G_{n;x_j}(\xi,\mu) = \frac{1}{2 \pi}{\rm e}^{\frac{i}2 (\psi_n(\xi;2t + 2x_j) - \psi_n(\mu;2t + 2 x_j))}\chi_{\Gamma_+(\xi)}\chi_{\Gamma_-(\mu)}.
		$$
		We consider the  two corresponding operators $\mathbb{F}_n, \mathbb{G}_{n;x_j}$ extended to the whole space $ L^2(\Gamma) = L^2(\Gamma_+) \oplus L^2(\Gamma_-) $, acting trivially on the respective orthogonal component. We first notice that both the operators $\mathbb{F}_n, \mathbb{G}_{n;x_j}$ are Hilbert-Schmidt on the whole space. Indeed: 
		\begin{equation}
		   \big\| \mathbb{F}_n \big\|_2 ^2 = \frac{1}{(2\pi)^2}\int_{\Gamma_+}\vert \d \lambda \vert \int_{\Gamma_-} \vert \d \mu\vert \frac{{\rm e}^{-\mathfrak{I} (\psi_n(\lambda;0) - \psi_n(\mu;0))}}{\vert\mu - \lambda\vert^2} < +\infty
		\end{equation}
		and also 
		\begin{equation}
		    \big\| \mathbb{G}_{n;x_j}\big\|_2 ^2 = \frac{1}{(2\pi)^2}\int_{\Gamma_-}\vert  \d \mu \vert  \int_{\Gamma_+} \vert \d \xi \vert \frac{{\rm e}^{-\mathfrak{I} (\psi_n(\xi;2t+2x_j) - \psi_n(\mu;2t + 2x_j))}}{\vert\xi - \mu\vert^2} < +\infty.
		\end{equation}
	Moreover, they are both trace-class, since they both can be obtained as composition of Hilbert-Schmidt operators. To see that, we consider a new contour $\Gamma_0 \coloneqq \mathbb{R}+ \epsilon$, not intersecting either $\Gamma_+$ and $\Gamma_-$. 
	We start by the case of $\mathbb{G}_{n;x_j}$. We define the following two operators 
	\begin{equation}
	\begin{aligned}
	\mathbb{B}_{n;x_j}^{(1)}: L^2(\Gamma_-)\to L^2(\Gamma_0), \quad \text{with kernel}\quad B_{n;x_j}^{(1)} (\zeta,\mu) = \frac{{\rm e}^{-\frac{\i}{2}\psi_n(\mu;2x_j+2t)} }{2\pi \i (\zeta-\mu)} &\\
		\mathbb{B}_{n;x_j}^{(2)}: L^2(\Gamma_0)\to L^2(\Gamma_+), \quad \text{with kernel}\quad B_{n;x_j}^{(2)} (\lambda,\zeta) = \frac{{\rm e}^{\frac{\i}{2}\psi_n(\lambda;2x_j+2t)} }{2\pi (\zeta-\lambda )}.
	\end{aligned}
	\end{equation}
	Their composition gives $\mathbb{B}_{n;x_j}^{(2)} \circ \mathbb{B}_{x_j}^{(1)} = \mathbb{G}_{n;x_j}$ and, since they are both Hilbert-Schmidt, $\mathbb{G}_{n;x_j}$ is trace-class. 
	The case $\mathbb{F}_n$ is treated similarly using the two operators
	\begin{equation}
		\begin{aligned}
		\mathbb{C}^{(1)}_n:L^2(\Gamma_+)\to L^2(\Gamma_0)\quad \text{with kernel } \quad C^{(1)}_n (\lambda,\mu) = \frac{\rm{e}^{-\frac{\i}{2}}\psi_n(\mu;0)}{2\pi \i (\mu-\lambda)} &\\
		\mathbb{C}^{(2)}_n:L^2(\Gamma_0)\to L^2(\Gamma_-)\quad \text{with kernel } \quad C^{(2)}_n (\zeta,\lambda) = \frac{\rm{e}^{\frac{\i}{2}\psi_n(\zeta;0)}}{2\pi  (\lambda-\zeta)}.
		\end{aligned}
	\end{equation}
	This means, in particular, that the Fredholm determinants $\det(\I - \mathbb{L}_n)$ and $\det(\I-\mathbb{G}_{n;x_j})$ are well defined, and moreover $\det(\I-\mathbb{G}_{n;x_j}) \equiv 1$. Hence,
	\begin{eqnarray}
		{\rm det} ( \I - \mathbb{L}_n)&=&{\rm det} ( \I - \mathbb{G}_{n;x_j}) {\rm det} ( \I - \mathbb{L}_n) =\nonumber \\&=& {\rm det}\left( \begin{pmatrix} 
		\I & -\mathbb F_n  \\
		0 &	-\ds\sum_{j = 1}^k (\alpha_j - \alpha_{j + 1}) \mathbb G_{n;x_j} \circ \mathbb{F}_n
		\end{pmatrix}\right) =  {\rm det} \left(\I - \sum_{j = 1}^k (\alpha_j - \alpha_{j + 1})\mathbb G_{n;x_j}\circ \mathbb F_n\right), \nonumber
		\end{eqnarray}
		where the equality between the two lines is easily proven using the block representation of ${\mathbb L}_n$ and ${\mathbb G}_{n;x_j}$ induced by the polarization $L^2(\Gamma) \simeq L^2(\Gamma_+) \oplus L^2(\Gamma_-) $.\\
	Notice that each factor $\mathbb G_{n;x_j}\circ \mathbb F_n : L^2(\Gamma_+) \longrightarrow L^2(\Gamma_+)$ has kernel equal to
		$$
		(G_{n;x_j} \circ F_n)(\xi,\lambda) = \frac{1}{(2 \pi )^2}{\rm e}^{\frac{\i}2\big(\psi_n(\xi;2t + 2x_j) + \psi_n(\lambda;0) \big)} \int_{\Gamma_-} \frac{{\rm e}^{-\i \psi_n(\mu;t + x_j)}}{(\xi - \mu)(\mu - \lambda)} \d \mu.
		$$
		We now conjugate $\mathbb G_{x_j}\circ \mathbb F$ by the multiplication operator $\mathbb P_n$ with kernel $P_n(\lambda,\mu) = {\rm e}^{-\frac{i}2 \psi_n(\lambda,0)}\delta(\lambda - \mu)$ so to obtain
		$$
		(P_n \circ G_{x_j} \circ F \circ P_n^{-1})(\xi,\lambda) = \frac{1}{(2\pi)^2} {\rm e}^{\i(t + x_j)\xi} \int_{\Gamma_-} \frac{{\rm e}^{\i (\psi_n(\lambda;0) - \psi_n(\mu,x_j + t)}}{(\xi - \mu)(\mu - \lambda)} \d \mu.$$
		Next we observe that, in the double-contour integral representation \eqref{eq:doublecontourkernel} of $K_n$, one can deform the contour $\Gamma_+$ into $\mathbb R$. Using this property and  
		conjugating once more with the standard Fourier transform (i.e. with the integral operator with kernel $\mathcal F(x,\xi) = \frac{1}{\sqrt{2 \pi}}{\rm e}^{-x\xi}$) we finally obtain
		\begin{eqnarray}
			\Big(\mathcal F \circ P_n \circ G_{n;x_j} \circ F_n \circ P_n^{-1} \mathcal F^{-1}\Big)(x,y) = \\
			\frac{1}{(2 \pi)^2} \int_{\mathbb R} \frac{\d \xi}{\sqrt{2 \pi}} {\rm e}^{\i(t + x_j - x)\xi}\int_{\Gamma_-}\d \mu \int_{\mathbb R} \frac{\d \lambda}{\sqrt{2 \pi}} \frac{{\rm e}^{\i (\psi_n(\lambda;0) - \psi_n(\mu,x_j + t)}}{(\xi - \mu)(\mu - \lambda)} {\rm e}^{\i \lambda} = \\
			\begin{cases}
				\ds\frac{\i}{(2\pi)^2}\int_{\Gamma_-} \d \mu \int_{\mathbb R} \d \lambda \frac{{\rm e}^{\i(\psi_n(\lambda;y) - \psi_n(\mu,x))}}{\lambda - \mu} \quad &\text{if} \; x \geq t + x_j \\
				\\
				0 & \text{if} \; x < t + x_j 
			\end{cases}\quad\quad,
		\end{eqnarray}
		where the latter equality is obtained deforming the outer contour $\mathbb R$ toward $+\infty{\rm e}^{\pm \pi \i}$ (depending on the sign of $t + x_j - x$) and taking a residue.
		Since the conjugations we performed do not change the value of the Fredholm determinant, we find
		$$
		\det(\I - \mathbb L_n) = \det \left(\I - \sum_{j = 1}^k (\alpha_j - \alpha_{j+1}) \mathbb K_{n | [x_j + t,\infty)}\right),
		$$
		and this latter is equal to $F_n(\vec x + t,\vec\alpha)$ because of \eqref{eq:FredhDet}.
	\end{proof}
	As $\mathbb L_n$ is of integrable type, in the sense of \cite{IIKS}, it is naturally associated to a Riemann-Hilbert problem we are going to define, and whose jumps are given by the $k+1 $ dimensional matrix $J_{Y}(\lambda) := \mathbf 1_{k+1} - 2 \pi \i f_n(\lambda) g_n^\top(\lambda)$ for $\lambda \in \Gamma$.
	
	\begin{pb}\label{pb:main rhpb}
		Find a sectionally--analytic function $Y(\bullet;n,\x+t,\vec\alpha) : \mathbb C/ \Gamma \longrightarrow \mathrm{GL}(k+1,\mathbb C)$ such that
		\begin{itemize}
			\item[a)] $Y$ has continuous boundary values $Y_{\pm}$ as $\lambda \in \Gamma$ is approached from the left $(+)$ or right $(-)$ side, and they are related by
			\begin{equation}
				\begin{array}{rcl}
					Y_+(\lambda) &=& Y_-(\lambda) \left(\begin{array}{c|ccccc}
						1 & & & -\i \Theta_n^\top(-\lambda;\vec x +t, \vec \alpha)\chi_{\Gamma_-}(\lambda) & & \\
						\hline
						\\
						\\
						
						- \i \Theta_n(\lambda; \vec x + t, \vec\alpha) \chi_{\Gamma_+}(\lambda)&&& \mathbf 1_k &&
						\\
						\\
						\\
					\end{array}\right),
				\end{array}
			\end{equation}
			where $\Theta_n(\lambda; \vec x + t, \vec\alpha)$ is the (column) vector 
			\begin{equation}
				\Theta_n(\lambda; \vec x +t, \vec\alpha) := \left( \sqrt{\alpha_j - \alpha_{j+1}}{\rm e}^{\i \psi_n(\lambda;t + x_j)}\right)_{j = 1}^k.
			\end{equation}
			\item[b)] There exists a matrix $Y_1 = Y_1(n,\vec x + t, \vec \alpha)$, independent of $\lambda$, such that $Y$ satisfies
			\begin{equation}\label{eq:asymp RH}
				Y(\lambda) = \mathbf 1_{k+1} + Y_1 \lambda^{-1} + \mathcal O(\lambda^{-2}), \quad \lambda \to \infty.
			\end{equation}
		\end{itemize}
	\end{pb}
	We record in the following remark some symmetries that will be useful in the sequel.
	\begin{remark}\label{remark2.4}
		The jump matrix $J_Y(\lambda;n,\vec x + t,\vec \alpha) \equiv J_Y(\lambda)$ of the Riemann-Hilbert problem \ref{pb:main rhpb} satisfies the following two symmetries:
		$$
		J_Y^{-\top}(-\lambda) = D_1J_Y(\lambda)D_1 \quad \quad  \overline{J_Y(\overline{\lambda})} = D_2 J_Y(-\lambda) D_2,
		$$
		where $D_1 := \diag(1,-1\ldots,-1)$ and $D_2 := \diag(1,c_1,\ldots,c_k)$ with $c_j := -\sgn(\alpha_j - \alpha_{j+1})$. 
		Consequently, the unique solution $Y$ to the Riemann-Hilbert problem satisfy the symmetry relations
		\begin{equation}
			D_1Y^{\top}(-\lambda)D_1 = Y^{-1}(\lambda), \quad D_2\overline{Y(\overline{\lambda})}D_2 = Y(\lambda).
		\end{equation} 
		In particular, using the expansion of the function $Y$ at $\lambda \to \infty$ together with the symmetries above, we obtain that 
		\begin{equation}\label{Y1}
			Y_1 = \left(\begin{array}{c|ccccc}
				-\delta  & & \vec u^\top&  \\
				\hline
				\\
				\vec u&&\Delta &
				\\
				\\
			\end{array}\right)
		\end{equation}
		where the entry $u_j,\, j = 1, \ldots, k$ of $\vec u = \vec u(n,\vec x + t,\vec \alpha)$ is real if $\alpha_j - \alpha_{j + 1} > 0$ and purely imaginary if $\alpha_j - \alpha_{j + 1} < 0$. Moreover, since $\det Y \equiv 1$, we also have that $\delta = \mathrm{Tr}(\Delta)$.
	\end{remark}
	\begin{remark}
		The jump matrix $J_Y(\lambda;n,\vec x + t, \vec \alpha) \equiv J_Y(\lambda)$ can be factorized as
		\begin{equation}\label{factoriz jump}
			J_Y(\lambda) = \exp (M) J_Y(0) \exp (-M) 
		\end{equation}
		with $M\equiv M(\lambda;n,\vec x + t, \vec \alpha) \coloneqq \diag (M_0,M_1,\dots,M_k)$,
		\begin{equation}
			M_0 := -\frac{i}{k+1}\sum_{j=1}^{k}\psi_n(\lambda; t+ x_j),\;\;\;\; M_{\ell} := M_0 +i\psi_n(\lambda; x_{\ell}+t),\;\; \ell=1,\dots,k.
		\end{equation}
		Note that the matrix $ J_Y(0)$ does not depend on $\vec{x},t$ or $n$.
	\end{remark}

	\begin{proposition}\label{prop: first log derivative}
		The unique solution $Y$ of the Riemann--Hilbert problem \ref{pb:main rhpb} is related to the Fredholm determinant $F_n(\x;\vec{\alpha})$ via the formula
		\begin{equation}\label{eq: ds fred det}
			\frac{\partial}{\partial t} F_n(\x+t,\vec{\alpha}) = \i(Y_1)_{1,1}.
		\end{equation}
	\end{proposition}
	\begin{proof}
		As we proved that $F_n(\vec x + t,\vec \alpha)$ is the Fredholm determinant associated to the integrable kernel $L_n$ defined in  \eqref{eq:Lintegrable}, the following general formula, proven in \cite{BertolaCafasso1}, Theorem 3.2, relates $F_n(\x + t,\vec\alpha)$ to the (unique) solution of the Riemann--Hilbert problem \ref{pb:main rhpb}:
	\begin{equation}
			\frac{\partial}{\partial t } F(\vec{x}+t,\vec{\alpha}) = \int_{\Gamma}\tr \left( Y^{-1}_-(\lambda) Y_-'(\lambda) \frac{\partial}{\partial t }J_Y(\lambda) J_Y^{-1}(\lambda)\right)\frac{\ d\lambda}{2\pi \i},
		\end{equation}
		where the symbol  $'$ denotes the derivative w.r.t. the complex parameter $\lambda$.
		Thanks to the factorization of the jump matrix $J_Y(\lambda)$ given in \eqref{factoriz jump} and the nature of the contour $\Gamma = \Gamma_+ \cup \Gamma_-$, the integral in the right hand side is computed as a formal residue at infinity. Indeed, we start noticing that
		\begin{equation}
		\begin{aligned}
		 	\int_{\Gamma_+\cup\Gamma_-}\tr \left( Y^{-1}_-(\lambda) Y_-'(\lambda) \frac{\partial}{\partial t }J_Y(\lambda) J_Y^{-1}(\lambda)\right)\frac{\d\lambda}{2\pi \i}& =  \\ \int_{\Gamma_+\cup\Gamma_-}\tr \left( Y^{-1}_-(\lambda) Y_-'(\lambda) \left( \frac{\partial}{\partial t }M(\lambda) - J_Y(\lambda)\frac{\partial}{\partial t }M(\lambda) J_Y^{-1}(\lambda)\right)  \right)\frac{\d\lambda}{2\pi \i} =\\
		 	\underbrace{\int_{\Gamma_+\cup\Gamma_-}\tr \left( Y^{-1}_-(\lambda) Y_-'(\lambda) \frac{\partial}{\partial t }M(\lambda) \right)\frac{\d\lambda}{2\pi \i}}_{(\clubsuit)} -  \int_{\Gamma_+\cup\Gamma_-}\tr \left( Y^{-1}_+(\lambda) Y_+'(\lambda)\frac{\partial}{\partial t }M(\lambda)\right)\frac{\d\lambda}{2\pi \i} 
		 	\\ + \underbrace{\int_{\Gamma_+\cup\Gamma_-}\tr \left(J_Y^{-1}(\lambda) J_Y'(\lambda)\frac{\partial}{\partial t }M(\lambda) \right)\frac{\d\lambda}{2\pi \i}}_{(\spadesuit)} =\\
		 	-\lim_{R\rightarrow + \infty}\int_{C_R} \tr \left( Y^{-1}(\lambda) Y'(\lambda)\frac{\partial}{\partial t }M(\lambda)\right)\frac{\d\lambda}{2\pi \i}. 
		\end{aligned}
		\end{equation}
		In the last passage we used that the terms denoted with $(\clubsuit)$ and $(\spadesuit)$ are zero (as we will see in a moment) and we rewrote the remaining term by deforming the contour $\Gamma_+ \cup \Gamma_-$ into a circle $C_R$ centered in zero and of increasing radius $R$. Indeed, $(\spadesuit)$ is zero because of the form of the jump matrix $J_Y$ (we are computing the trace of a strictly lower/upper triangular matrix), while $(\clubsuit)$ is zero because the integrations along $\Gamma_+$ and $\Gamma_-$ cancel out. Finally, using the asymptotic expansion of $Y$ at infinity combined with 
		$$\frac{\partial M}{\partial t }=\frac{i \lambda}{k+1}\diag(-k,1,\dots,1),$$
		we explicit compute
		$$
			-\lim_{R\rightarrow + \infty}\int_{C_R} \tr \left( Y^{-1}(\lambda) Y'(\lambda)\frac{\partial}{\partial t }M(\lambda)\right)\frac{\d\lambda}{2\pi \i} =  \i(Y_1)_{1,1}.
		$$
	\end{proof}
	This result will be used in the last section to relate $F_n(\vec{x}+t, \vec{\alpha})$ to a distinguished solution of the vector--valued PII hierarchy, whose Lax pair is given in the following section.
	\section{A Lax pair associated to a vector valued PII hierarchy}
	\label{sec3}
	We now introduce a new matrix $\Psi(\lambda) = \Psi(\lambda;n,\vec x + t,\vec\alpha)$ solving a Riemann-Hilbert problem with constant jumps. More specifically, let
	\begin{eqnarray}\label{def:T}
		T^{(1)}(\lambda;n,\vec x+t) &:=& \diag\left(1,{\rm e}^{\i\psi_n(\lambda;x_1 + t)}, \cdots, {\rm e}^{\i\psi_n(\lambda;x_k + t)}\right), \\ 
		T^{(2)}(\lambda;n,\vec x+t) &:=& {\rm e}^{-\frac{\i}{k+1} \sum_{j = 1}^k \psi_n(\lambda;x_j + t)}\mathbf 1_{k+1}
	\end{eqnarray}
	and $T(\lambda) \equiv T(\lambda;n, \vec x+t) := T^{(1)}(\lambda; n, \vec x + t)T^{(2)}(\lambda; n, \vec x +t)$. We define
	\begin{equation}\label{def:Psi}
		\Psi(\lambda;n,\vec x + t, \vec\alpha) := Y(\lambda;n,\vec x + t,\vec\alpha) T(\lambda; n,\vec x + t).
	\end{equation} 
	It is easy to prove that this latter satisfies the following Riemann--Hilbert problem.
	\begin{pb}\label{RHPPsi}
		Find a sectionally--analytic function $\Psi : \mathbb C/ \Gamma \longrightarrow \mathrm{GL}(k+1,\mathbb C)$ such that
		\begin{itemize}
			\item[a)] $\Psi$ has continuous boundary values $\Psi_{\pm}$ as $\lambda \in \Gamma$ is approached from the left $(+)$ or right $(-)$ side, and they are related by
			\begin{equation}
				\begin{array}{rcl}
					\Psi_+(\lambda) &=& \Psi_-(\lambda) \left(\begin{array}{c|ccccc}
						1 & & & -\i \hat\Theta^\top\chi_{\Gamma_-}(\lambda) & & \\
						\hline
						\\
						\\
						
						- \i \hat\Theta\chi_{\Gamma_+}(\lambda)&&& \mathbf 1_k &&
						\\
						\\
						\\
					\end{array}\right),
				\end{array}
			\end{equation}
			where $\hat\Theta := \Theta(\lambda;n,\vec x + t, \vec\alpha)_{| \lambda = 0}$ is the (column) vector 
			\begin{equation}
				\hat\Theta := \Big(\sqrt{\alpha_j - \alpha_{j+1}}\Big)_{j = 1}^k.
			\end{equation}
			\item[b)] $\Psi$ has the following asymptotic behavior 
			\begin{equation}
				\Psi(\lambda) = \Big( \mathbf 1 + Y_1 \lambda^{-1} + \mathcal O(\lambda^{-2}) \Big) T(\lambda), \quad \text{as} \; \lambda \to \infty.
			\end{equation}
		\end{itemize}
	\end{pb}
	The properties of the function $\Psi(\lambda)$ listed above are then used to prove the following proposition.
	\begin{proposition}
		There exists two matrices $A(\lambda) \equiv A(\lambda;n,\vec x + t,\vec\alpha)$ and $B(\lambda) \equiv B(\lambda;n, \vec x + t, \vec\alpha)$, polynomial in $\lambda$, such that
		\begin{equation}\label{LaxPair}
			\begin{cases}
				\ds\frac{\partial}{\partial \lambda} \Psi(\lambda) = A(\lambda)\Psi(\lambda),\\
				\\
				\ds\frac{\partial}{\partial t} \Psi(\lambda) = B(\lambda)\Psi(\lambda).
			\end{cases}
		\end{equation}
		Moreover,
		\begin{equation}\label{eq:B}
			B(\lambda) = \left(\begin{array}{c|ccccc}
				-\ds\frac{\i k}{k+1}\lambda  & & -\i \vec u^\top&  \\
				\hline
				\\
				\i \vec u&&\ds\frac{\i \lambda}{k+1} \mathbf 1_k &
				\\
				\\
			\end{array}\right)
		\end{equation}
		and
		\begin{equation}\label{eq:A}
			A(\lambda) = \sum_{j = 0}^{2n} A_j \lambda^{2n - j} + \hat A_{2n},
		\end{equation}
		with
		\begin{eqnarray}\label{eq:leading last coeff A}
			A_0 &=& \frac{\i}{k+1}\diag\left(- k ,  1 , \ldots, 1 \right),\\ \hat A_{2n} &=& \frac{\i}{k+1}\diag\left(-kt- \sum_{j = 1}^k x_j,t + k x_1 - \sum_{j \neq 1} x_j , \ldots, t +   k x_k - \sum_{j \neq k} x_j \right) \label{eq:leading last coeff Abis}.
		\end{eqnarray}
		and $\vec u$ as in \eqref{Y1}.
	\end{proposition} 
	\begin{proof}
		For both equations the proof follows applying Liouville theorem and using straightforward computations. For the first equation, we define 
		\[A(\lambda) \coloneqq\ds\frac{\partial}{\partial \lambda} \Psi(\lambda)\left(\Psi(\lambda)\right)^{-1}. \]
		The matrix-valued function $A(\lambda)$ is analytic for every $\lambda \in \mathbb{C}\setminus \Gamma $. Moreover, for $\lambda \in \Gamma$ we have that $A_+(\lambda) = A_-(\lambda) $, thanks to the fact that $\Psi(\lambda)$ has constant jump condition along $\Gamma$. Thus $A(\lambda) $ is entire and behaves like a polynomial in $\lambda$ of degree $2n$ at $\infty$. By Liouville theorem, we conclude that $A(\lambda) $ is a polynomial of degree $2n$ in $\lambda$. Using the asymptotic condition written in equation \eqref{eq:asymp RH}, we can then compute explicitely the leading coefficient and the constant coefficient of $A(\lambda)$ as in \eqref{eq:leading last coeff A}, \eqref{eq:leading last coeff Abis}. For the second equation, in an analogue way we define 
		\[B(\lambda) \coloneqq\ds\frac{\partial}{\partial t} \Psi(\lambda)\left(\Psi(\lambda)\right)^{-1}. \]
		Using the same reasoning, we conclude that $B(\lambda)$ is a polynomial in $\lambda$ of degree $1$ and we compute its coefficients as in \eqref{eq:B}.
	\end{proof}
	Now, we show that the system \eqref{LaxPair} is the Lax pair for the vector-valued Painlevé II hierarchy \eqref{PIIequationIntro}.
	For any $j = 0, \ldots, 2n$, we will denote with $a_j^{11}, a_j^{12}, a_j^{21}, a_j^{22}$ the block-entries of $A_j$ (and the same for $\hat A_{2n}$). Hence, $a_j^{11}$ will be a scalar, $a_j^{12}, a_j^{21}$ will be, respectively, a row and a column vector, and $a_j^{22}$ a square matrix of size $k$. We will now study the compatibility condition
	\begin{equation}\label{LaxEq}
		A(\lambda)B(\lambda) - B(\lambda)A(\lambda) = \frac{\partial B}{\partial \lambda}(\lambda) - \frac{\partial A}{\partial t}(\lambda).
	\end{equation}
	The following two Lemmas are the analogue, for the vector--valued case, of Lemma 5.4 and 5.5 in \cite{BothnerCafassoTarricone} and the tecnique used in their proofs is inspired by the one used in \cite{WE}. The dependence of the variable $\vec u$ and $\{ a_j^{ik} \}$ on $n,t,\vec x, \vec \alpha$ will not be made explicit in the following formulas. We will denote with a dot the derivative with respect to $t$, as this will be the only dynamical variable (the other variables will be considered as parameters). 
	\begin{lemma}\label{lemma:1}
		The compatibility condition of the Lax pair \eqref{LaxPair} is equivalent to the system of equations
		\begin{equation}\label{compatibility1}
			a_{1}^{12} = -\i \vec u^\top, \quad a_1^{21} = \i \vec u,
		\end{equation}
		\begin{equation}\label{compatibility2}
			\begin{cases}
				\dot{a}_j^{11} = -\i(\vec u^\top a_j^{21} + a_j^{12}\vec u), & \dot{a}_{j}^{12} = -\i(a_{j+1}^{12} + \vec u^\top a_j^{22} - a_j^{11}\vec u^\top) \\
				\\
				\dot{a}_j^{22} = \i(\vec ua_j^{12} + a_j^{21}\vec u^\top), & \dot{a}_{j}^{21} = \i(a_{j+1}^{21} + a_j^{11} \vec u - a_j^{22}\vec u)
			\end{cases}, \quad j = 1,\ldots,2n-1
		\end{equation}
		\begin{equation}\label{compatibility3}
			\begin{cases}
				\dot{a}_{2n}^{11} = -\i(\vec u^\top a_{2n}^{21} + a_{2n}^{12}\vec u), & \dot{a}_{2n}^{12} = -\i(\vec u^\top a_{2n}^{22} - a_{2n}^{11}\vec u^\top + \i \vec u^\top m_{t,\x}) \\
				\\
				\dot{a}_{2n}^{22} = \i(\vec ua_{2n}^{12} + a_{2n}^{21}\vec u^\top), & \dot{a}_{2n}^{21} = \i(a_{2n}^{11} \vec u - a_{2n}^{22}\vec u - \i m_{t,\x}\vec u),
			\end{cases}
		\end{equation}
		where $m_{t, \x} := \diag(x + t_1,\ldots,x + t_k)$.
	\end{lemma}
	\begin{proof}
		The proof follows by direct computation, exploiting the polynomiality of the matrices $A(\lambda), B(\lambda)$ in $\lambda$. In particular, equation \eqref{compatibility1} corresponds to the term associated to the monomial $\lambda^{2n}$ in the compatibility condition \eqref{LaxEq}. Then system \eqref{compatibility2} corresponds, for every $j=1,\dots,2n-1$, to the monomial $\lambda^{2n-j}$ and, finally, \eqref{compatibility3} corresponds to $\lambda^{0}.$
	\end{proof}
	From equations \eqref{compatibility1}, \eqref{compatibility2}, and the first two equations of \eqref{compatibility3} we can derive some symmetries for the entries of $A_j$ for $j=1,\dots, 2n$. In particular,
	\[\frac{\partial }{\partial t} \tr (a_j^{22})=\tr(\dot{a}_j^{22}) = \i \tr (\vec u a_j^{12} +a_j^{21} \vec u^{\top} ) = \i \tr (\vec u a_j^{12}) +i \tr(a_j^{21} \vec u^{\top} )=  \i a_j^{12} \vec u +i \vec u^{\top} a_j^{21}   = -\dot{a}_j^{11}.\]
	Thus we conclude that 
	\begin{equation}\label{eq:Trace}
		-\tr (a_j^{22}) =  a_j^{11}, \;\; j=1\dots,2n,
	\end{equation}
	up to constant of integration. This constant is actually zero. This can be proven observing that $\lim_{t \to \infty} a_j^{22} = \lim_{t \to \infty} a_j^{11} = 0$. Indeed, both of them are polynomials in the entries of the matrix elements of $Y_i, \; i \geq 1$, and  $Y_i\rightarrow 0 $ for $t\rightarrow +\infty$, which is easily proven using the small norm theorem (see Section \ref{sec:4} below and, in particular, the proof of Proposition \ref{prop:4.1}).	
	Moreover, we can also deduce the following symmetries : 
	\begin{equation}\label{eq:symmetries}
		a^{21}_j = (-1)^j (a^{12}_j)^\top, \;\; a^{22}_j = (-1)^j (a^{22}_j)^{\top} \;\;\; \forall j = 1,\ldots, 2n.
	\end{equation}
	This is proved by induction over $j$ using the equations appearing in the Lemma \ref{lemma:1}.
	\begin{lemma}\label{lemma:2}
		For every $\ell = 1,2,\ldots,2n$
		\begin{equation}
			a_\ell^{11} = -\i \sum_{j = 1}^{\ell - 1} \big(a_j^{11}a_{\ell - j}^{11} + a_j^{12}a_{\ell - j}^{21}\big) \quad \text{and} \quad a_\ell^{22} = \i \sum_{j = 1}^{\ell - 1}\big(a_j^{22}a_{\ell - j}^{22} + a_j^{21}a^{12}_{\ell - j}\big).
		\end{equation}
	\end{lemma}
	\begin{proof}
		Thanks to equation \eqref{eq:Trace}, we only need to prove the formula for $a_\ell^{22}$. Consider the matrix
		\[C \coloneqq A^2 =
		\left(\begin{array}{c|ccc}
			-\frac{k^2}{(k+1)^2}& & 0 &\\
			\hline
			&&&\\
			0^{\top} & & -\dfrac{\1_k}{(k+1)^2} &\\
			&&&
		\end{array}\right)\lambda^{4n} + \sum_{\ell=1}^{4n}\lambda^{4n-\ell}C_{\ell},\]
		where $C_{\ell} = \sum_{j = 0}^{\ell}A_j A_{\ell-j}$ for $\ell=1,\dots, 2n-1$ and $C_{2n} = \sum_{j = 0}^{2n}A_j A_{2n-j}+A_0\hat{A}_{2n}+\hat{A}_{2n}A_0$ (these are the only coefficients we are going to use in the proof). We can write $C_{\ell}$, for every value of $\ell$, in the usual block-form 
		\[C_{\ell} \coloneqq  \left(\begin{array}{c|ccc}
			c_{\ell}^{11} & & c_{\ell}^{12} &\\
			\hline
			&&&\\
			c_{\ell}^{21} & & c_{\ell}^{22}  &\\
			&&&
		\end{array}\right),\]
		and, in particular, each block-entry can be written in terms of the block entries of  the matrix $A$
		\begin{equation}\label{definition entries of C}
			\begin{cases}
				c_{\ell}^{11} = \sum_{j = 0}^{\ell} \left( a_j^{11} a_{\ell-j}^{11} + a_j^{12}a_{\ell-j}^{21} \right) +  2a_0^{11}\hat{a}_{2n}^{11}\delta_{\ell, 2n}\\
				c_{\ell}^{22}= \sum_{j = 0}^{\ell}  \left( a_j^{21} a_{\ell-j}^{12} + a_j^{22}a_{\ell-j}^{22} \right)+ (a_0^{22}\hat{a}_{2n}^{22}+ \hat{a}_{2n}^{22}a_0^{22})\delta_{\ell, 2n}\\ 
				c_{\ell}^{12}= \sum_{j = 0}^{\ell}  \left( a_j^{11} a_{\ell-j}^{12}+ a_j^{12}a_{\ell-j}^{22} \right)\\
				c_{\ell}^{21}= \sum_{j = 0}^{\ell} \left( a_j^{21} a_{\ell-j}^{11} + a_j^{22}a_{\ell-j}^{21} \right)
			\end{cases} \quad \ell = 1,\dots,2n.
		\end{equation}
		From the compatibility condition \eqref{LaxEq}, we deduce the following equation for the matrix $C$
		\begin{equation}\label{newlaxeq}
			\frac{\partial C}{\partial t}(\lambda) = B(\lambda) C(\lambda) - C(\lambda) B(\lambda) + \frac{\partial B}{\partial \lambda}(\lambda) A(\lambda) + A(\lambda)\frac{\partial B}{\partial \lambda}(\lambda).
		\end{equation}
		In particular, by looking at the coefficients of the powers $\lambda^{m}$ with $m=4n,\dots, 2n$ in equation \eqref{newlaxeq}, we obtain the following system of difference and differential equations for the block entries of each coefficient $C_{\ell}$
		\begin{equation}\label{diff eq for C}
			\begin{cases}
				\dot{c}_{\ell}^{11} = -\i(\vec u^\top c_{\ell}^{21} + c_{\ell}^{12}\vec u) + 2(a_0^{11})^2\delta_{\ell,2n}, & \dot{c}_{\ell}^{12} = -\i(c_{\ell+1}^{12} + \vec u^\top c_{\ell}^{22} - c_{\ell}^{11}\vec u^\top) \\
				\\
				\dot{c}_{\ell}^{22} = \i(\vec u c_{\ell}^{12} + c_{\ell}^{21}\vec u^\top) + 2(a_0^{22})^2 \delta_{\ell, 2n}, & \dot{c}_{\ell}^{21} = \i(c_{\ell+1}^{21} + c_{\ell}^{11} \vec u - c_{\ell}^{22}\vec u)
			\end{cases}, \quad \ell = 1,\ldots,2n.  
		\end{equation}
		Note that this is almost the same system satisfied by the matrix  elements of $A$, see \eqref{compatibility2}. In particular, the equations giving the entries $(1,2), (2,1)$ in \eqref{compatibility2} and \eqref{diff eq for C} are exactly the same, while for the entries $(2,2), (1,1)$ the two sets of equations differ just for $\ell=2n$.


        
        Now, by using the equations above, we first prove that the entries of $C_{\ell}$, for $\ell=1,\dots,2n$ are multiple of the ones of $A_{\ell}$, by induction over $\ell.$ Keeping in mind that the coefficient $C_0$ is explicitely written in the definition of the matrix $C$, we start by computing $C_1$:
		\begin{equation}
			c_1^{12} =\i \frac{1-k}{k+1}  a_1^{12}, \;\;  c_1^{21} =\i\frac{1-k}{k+1}a_1^{21}, \;\; c_1^{11} = 0=\i\frac{1-k}{k+1}a_1^{11}, \;\; c_1^{22} =0_k=\i\frac{1-k}{k+1}a_1^{22}. 
		\end{equation}
		We suppose now that the equation above holds for $\ell$ and we prove than it holds for $\ell+1$ too, by using the equations \eqref{diff eq for C} together with equations \eqref{compatibility1}, \eqref{compatibility2}. In particular, from the third equation in the system \eqref{diff eq for C} we recover $c_{\ell+1}^{12}$ as
		\begin{equation}
			c_{\ell+1}^{12} = \i  \dot{c}_{\ell}^{12} -\vec u^{\top}  c_{\ell}^{22}+ c_{\ell}^{11}\vec u^{\top}= \i \frac{1-k}{k+1}\underbracket{\left( \i \dot{a}_{\ell+1}^{12}-\vec u^{\top}a_{\ell}^{22}+a_{\ell}^{11} \vec u^{\top}\right)}_{=a_{\ell+1}^{12}},
		\end{equation}
		 using the induction hypothesis and the equation for $a_{\ell+1}^{12}$ in \eqref{compatibility2}. The same procedure can be applied for the block entry $c_{\ell+1}^{21},$ to get the analogue result. Now, for the diagonal entries, we use instead the first two equations in the system \eqref{diff eq for C} and obtain
		\begin{equation}
			\dot{c}_{\ell+1}^{22} =  \i\left( \vec u c_{\ell+1}^{12}+c_{\ell+1}^{21}\vec u^{\top}\right) = \i \frac{1-k}{k+1}\underbracket{\i \left( \vec u a_{\ell+1}^{12}+a_{\ell+1}^{21}\vec u^{\top}\right) }_{=\dot{a}_{\ell+1}^{22}}.
		\end{equation}
		Thus we conclude that $c_{\ell+1}^{22} =\i \frac{1-k}{k+1}a_{\ell+1}^{22}$ (from the equation above, this is true up to  a constant of integration, that can be fixed to zero thanks to equations \eqref{definition entries of C}). The same can be done for $c_{\ell+1}^{11}$. Thus the proportionality relation between the block entries of type $(1,2) $ or $(2,1)$ of $A_{\ell}$ and of $C_{\ell}$  is proved for $\ell=1\dots,2n-1$, while for the diagonal block entries it holds for  $\ell=1\dots,2n-2$. For $\ell =2n-1$ we have for the diagonal block entries
		\begin{equation}
		\dot{c}_{2n}^{22} = \i \frac{1-k}{k+1}\dot{a}_{2n}^{22} +2 (a_0^{22})^2, 
		\end{equation}
		and an analogue relation for $a_{2n}^{11}$.
		Now, by using the first two equations in the system \eqref{definition entries of C} for $c_{\ell}^{22} $, we deduce the following chain
		\begin{equation}
			\i\frac{1-k}{k+1}a_{\ell}^{22}= c_{\ell}^{22}=\sum_{j = 1}^{\ell-1}  \left( a_j^{21} a_{\ell-j}^{12} + a_j^{22}a_{\ell-j}^{22} \right)+\underbracket{a_0^{21} a_{\ell}^{12} + a_0^{22}a_{\ell}^{22}+a_{\ell}^{21} a_{0}^{12} + a_{\ell}^{22}a_{0}^{22}}_{= \i \frac{2}{k+1} a_{\ell}^{22}}
		\end{equation}
		for $\ell =1,\dots,2n-1$, from which the statement for $a_{\ell}^{22}$ is directly obtained. For the case $\ell=2n$ the chain obtained by replacing the equation  \eqref{definition entries of C} for $c_{2n}^{22} $ is the following one
		\begin{equation}
		\i\frac{1-k}{k+1}\dot{a}_{2n}^{22}+ 2(a_0^{22})^2 = \dot{c}_{2n}^{22} =  \dfrac{d}{dt}\left(\sum_{j = 0}^{\ell}  \left( a_j^{21} a_{\ell-j}^{12} + a_j^{22}a_{\ell-j}^{22} \right)\right) + 2(a_0^{22})^2.
		\end{equation}
		Thus, by simplifying both sides the term $2\big(a_0^{22}\big)^2$ and then integrating, the formula for $a_{2n}^{22}$ is obtained as for the previous values of $\ell$.
	\end{proof}
	Combining the two lemmas \ref{lemma:1} and \ref{lemma:2}, we are able to express all the coefficients of the Lax matrix $A(\lambda)$ in function of the vector $u$ and its derivatives.
	\begin{proposition}
		The entries of the matrix $A(\lambda)$ are differential polynomials on $u$, given recursively by the formulas
		\begin{equation}
			a_{1}^{21} = \i \vec u, \quad a_1^{22} = 0,
		\end{equation}
		\begin{equation} \label{recursion}
			a_{j + 1}^{21} = -\i \dot{a}_j^{21} - a_j^{11} \vec u + a_j^{22}\vec u, \quad a_{j+1}^{22} = \i \sum_{\ell = 1}^j \big(a_\ell^{22}a_{j + 1 -\ell}^{22} + a_\ell^{21}a^{12}_{j + 1 - \ell}\big), \quad j=1,\dots,2n
		\end{equation}
		together with \eqref{eq:Trace} and \eqref{eq:symmetries}.
	\end{proposition}
	

	Moreover, using \eqref{recursion} together with \eqref{compatibility3} and the definition of the operators $\mathcal L_\pm^{\vec u}$, we prove that the vectors $a_j^{21}, \quad j = 1, \ldots,2n$ satisfy the recursion
	\begin{equation}\label{recursion2}
		a^{21}_{2j + 1} = -\mathcal L_+^{\vec u} a^{21}_{2j}, \quad a^{21}_{2j} =  -\mathcal L_-^{\vec u} a^{21}_{2j-1}, \quad j = 1,\ldots, n-1,
	\end{equation}
	while $a^{21}_{2n}$ satisfies the differential equation
	\begin{equation}
		-\mathcal L_+^{\vec u} a_{2n}^{21}  = -\i m_{t,\x}\vec u.
	\end{equation}
	Hence, we proved recursively that $\vec u$, as defined in \eqref{Y1}, satisfies the equation \eqref{PIIequationIntro}. Moreover, the entries of $\vec u$ are real/purely imaginary depending on the sign of $(\alpha_{j + 1} - \alpha_j)$, as already observed in Remark \ref{remark2.4}. We are now left with the proofs of the equations \eqref{TWformula} and \eqref{asymptoticsIntro}.
	
	\section{The logarithmic derivative of $F_n(\x,\vec\alpha)$}\label{sec:4}
	In this last section we finish the  proof of Theorem \ref{thm:main} giving the relation between $F_n(\x,\vec\alpha)$ and a particular solution $\vec u$ of the vector-valued Painlevé II hierarchy. The main ingredient is the first logarithmic derivative of the Fredholm determinant $F_n(\vec{x}+t,\vec{\alpha})$, computed at the end of Section \ref{sec:2}. 
	
	\begin{proposition}\label{prop:4.1}
	The distribution $F_n$ defined in \eqref{def:genfunction} satisfies the equation
		\begin{equation}\label{second log der}
			\frac{\partial ^{2}}{\partial t ^{2}}F_n(\vec{x}+t, \vec{\alpha}) = - <\vec u(t), \vec u(t)>
		\end{equation}
		and its integrated version
		\begin{equation}\label{formula for fred det}
			F_n(\vec x, \vec \alpha) =\exp\left( - \int_{0}^\infty t <\vec u(t),\vec u(t)> \mathrm d t \right),
		\end{equation}
		where $\vec u(t) = \vec u(n,\vec x + t,\vec \alpha)$ satisfy the following (vector-valued) ordinary differential equation
		\begin{equation}\label{PIIequation}
			(\mathcal L^{\vec u}_+\mathcal L^{\vec u}_-)^n \vec u(t) = -\diag(x_1 + t,\ldots,x_k +t) \vec u(t)
		\end{equation}
		and have the following behavior at $+ \infty$
		\begin{equation}\label{asymptotics}
			\vec u(n,\vec x + t, \vec \alpha) = \Bigg(\sqrt{\alpha_j - \alpha_{j+1}}\Ai_n(t + x_j)\big(1 + o(1)\big)\Bigg)_{j = 1,\ldots,k}.
		\end{equation}
	\end{proposition}
	\begin{proof}
		We start computing the $\lambda^{-1}$--term in the asymptotic expansion of $(\partial_t \Psi)\Psi^{-1}$ for $\lambda \to \infty$,  where $\Psi$ is the solution of the Riemann-Hilbert problem \ref{RHPPsi} with constant jump condition. The $(1,1)$--entry of this term, which is equal to zero because $B(\lambda)$ is polynomial in $\lambda$, leads to
		$$ \frac{\partial}{\partial t }(Y_1)_{11} - \i \vec u^{\top}\vec u = 0$$
		and this equation, together with \eqref{eq: ds fred det}, gives \eqref{second log der}. 
		
		We describe now  the boundary behavior of $\vec u$ for $t \rightarrow +\infty$. We start by proving that the jump matrix $J_Y(\lambda)$ of the Riemann--Hilbert problem \ref{pb:main rhpb}, for $t \to +\infty$, behaves like the identity matrix, so that the small norm theorem can be applied. We consider a rescaled complex variable $w$ defined through the equation $\lambda=wt^{\frac{1}{2n}}$, so that the entries $(1,j+1)$ and $(j+1,1)$ of $J_Y(\lambda)$ for $j=1,\dots,k$ are rewritten respectively as
		\begin{equation}
			\begin{aligned}
				&-\i \left(\Theta (\vec{x}+t,\vec{\alpha}, -\lambda)\right)_j\chi_-(\lambda) = -\sqrt{\alpha_j-\alpha_{j+1}} \e ^ {-\i t^{\frac{2n+1}{2n}}\big(\frac{1}{2n+1}w^{2n+1}+(1 +\frac{x_j}{t})w\big)}\chi_{\Gamma_-}\big(wt^{\frac{1}{2n}}\big),\\
				&-\i\left(\Theta (\vec{x}+t,\vec{\alpha},\lambda)\right)_j\chi_+(\lambda) = -\sqrt{\alpha_j-\alpha_{j+1}} \e ^ {\i t^{\frac{2n+1}{2n}}\big(\frac{1}{2n+1}w^{2n+1}+(1 +\frac{x_j}{t})w\big)}\chi_{\Gamma_+}\big(wt^{\frac{1}{2n}}\big).
			\end{aligned}
		\end{equation}
		Note that the quantity $d_j\coloneqq 1+\frac{x_j}{t}$ is bounded in the regime $t\rightarrow \infty$ for any fixed $x_j$, as it converges to $1$. Thus we can modify the curves $\Gamma_-$ and $\Gamma_+$ into $\tilde{\Gamma}_-$ and $\tilde{\Gamma}_+$ (as done in \cite{CafassoClaeysGirotti}, Section 3) so  that 
		\begin{equation}
			\mathfrak{I} \left( \frac{w^{2n+1}}{2n+1}+d_j w\right)< 0, \;\;\; \text{for } \;\;\; w \in \tilde{\Gamma}_-\;\;\text{and} \;\;\mathfrak{I} \left( \frac{w^{2n+1}}{2n+1}+d_j w\right)> 0, \;\;\; \text{for } \;\;\; w \in \tilde{\Gamma}_+.
		\end{equation}
		In this way we obtain 
		\begin{equation}
			\big \| J_Y(wt^{\frac{1}{2n}}) - \1_{k+1} \big\|_{\infty} = \sqrt{\kappa_j-\kappa_{j+1}}\sup_{w\in \tilde{\Gamma}_{\pm}}\e ^ {\pm t^{\frac{2n+1}{2n}}\mathfrak{I}\left(\frac{w^{2n+1}}{2n+1}+d_j w\right)} \rightarrow 0
		\end{equation}
		for $t\rightarrow+ \infty$ and any fixed $x_j\in \mathbb{R}$. Thus, the rescaled function $X (w) \coloneqq Y \big(wt^{\frac{1}{2n}}\big)$, thanks to the  Riemann-Hilbert problem \ref{pb:main rhpb}, satisfies the following conditions:
		\begin{itemize}
			\item[--] it is analytic on $\mathbb{C}\setminus \tilde{\Gamma}_- \cup \tilde{\Gamma}_+$;
			\item[--] it admits continuous boundary values $X_{\pm}$ while approaching from the left or from the right of the curves $\tilde{\Gamma}_- \cup \tilde{\Gamma}_+$ and they are related through $X_+(w)=X_-(w)J_Y \big(wt^{\frac{1}{2n}}\big)$ for all $ w $ along the curves;
			\item[--] for $\vert w\vert \rightarrow \infty $ it behaves like the identity matrix, i.e. $X(w) \sim \1_{k+1} + \sum_{j\geq 1}\frac{X_j}{w^j}$.
		\end{itemize}
		Note that, in particular, $X_1 = t^{-\frac{1}{2n}}Y_1$ (we will use it in a moment). Moreover, by applying the small norm theorem (see for instance Theorem 5.1.5 in \cite{its2011large}) we can conclude that $X(w)\sim \1_{k+1}$ for $t\rightarrow+\infty$ and fixed finite $x_j$ and, in particular, $Y_j \to 0$ for $t \to + \infty$ (we already used this result in Section 3).\\
		
		On the other hand, because of its properties described above, $X$ satisfies the integral equation
		\begin{equation}
			X(w) = \1_{k+1} - \int_{\tilde{\Gamma}_+\cup\tilde{\Gamma}_-}X_-(v)\frac{f_n(vt^{\frac{1}{2n}}) g_n^{\top}(vt^{\frac{1}{2n}})}{v-w}dv.
		\end{equation}
		Expanding the right hand side for $w \to \infty$ we find
		\begin{equation}
			X_1 =  \int_{\tilde{\Gamma}_+\cup\tilde{\Gamma}_-}X_-(v)
			f_n( vt^{\frac{1}{2n}})g_n^{\top}(v t^{\frac{1}{2n}}) dv,
		\end{equation}
		and thus we can conclude that, for $t\rightarrow + \infty$ and fixed $x_j$, 
		\begin{equation}
			u_j(t;\vec{x}, \vec{\alpha}) =(Y_1)_{1,j+1}= t^{\frac{1}{2n}}(X_1)_{1,j+1}\sim  \frac{\sqrt{\alpha_j - \alpha_{j+1}}}{2\pi \i}\int_{\Gamma_+}\e ^ {\i\left(\frac{z^{2n+1}}{2n+1}+(t+x_j)z\right)}dz = \sqrt{\alpha_j - \alpha_{j+1}}\mathrm{Ai}_n(t+x_j)
		\end{equation}
		for any $j = 1, \ldots, k$, where in the last two equalities we used the small norm theorem and the integral representation of the $n$-th Airy function.
		Finally, the equation $\eqref{formula for fred det}$ is obtained integrating twice \eqref{second log der}, and checking that one can set the integration constants equal to zero because of the asymptotic behavior of $\vec u$ as $t \to +\infty$. 
	\end{proof}
	
\subsection*{Acknowledgements}

The authors are grateful to Thomas Bothner for the many useful discussions on the vector-valued Painlevé II hierarchy and its relation with the integro-differential one. We acknowledge the support of the H2020-MSCA-RISE-2017 PROJECT No. 778010 IPaDEGAN and the International Research Project PIICQ, funded by CNRS. S. T. was also supported by the Fonds de la Recherche Scientifique-FNRS under EOS
project O013018F.

\def\cprime{$'$}


\end{document}